\documentclass[conference]{IEEEtran}
\usepackage[boxruled]{algorithm2e}
\usepackage{cite}
\usepackage{graphicx}
\usepackage{psfrag}
\usepackage{subfigure}
\usepackage{url}
\usepackage{amsmath}
\usepackage{array}
\usepackage{amssymb}
\usepackage{amsfonts}
\usepackage{graphicx}
\usepackage{epstopdf}
\usepackage{algorithmic}

\newtheorem{lemma}{Lemma}

\newtheorem{corollary}{Corollary}

\newtheorem{definition}{Definition}
\newtheorem{theorem}{Theorem}

\newtheorem{example}{Example}

\usepackage{setspace}
\usepackage{amscd}
\usepackage{mathrsfs}
\usepackage{epsfig}
\usepackage{color}
\usepackage{textcomp}
\usepackage{multirow}
\usepackage{threeparttable}

\title{Linear Fractional Network Coding and Representable Discrete Polymatroids}
\begin{document}

\author{
\authorblockN{Vijayvaradharaj T. Muralidharan and B. Sundar Rajan}
\authorblockA{Dept. of ECE, Indian Institute of Science, Bangalore 560012, India, Email:{$\lbrace$tmvijay, bsrajan$\rbrace$}@ece.iisc.ernet.in
}
}

\maketitle
\begin{abstract}
A linear Fractional Network Coding (FNC) solution over $\mathbb{F}_q$ is a linear network coding solution over $\mathbb{F}_q$ in which the message dimensions need not necessarily be the same and need not be the same as the edge vector dimension. Scalar linear network coding, vector linear network coding are special cases of linear FNC. 
In this paper, we establish the connection between the existence of a linear FNC solution  for a network over $\mathbb{F}_q$ and the representability over $\mathbb{F}_q$ of discrete polymatroids, which are the multi-set analogue of matroids. All previously known results on the connection between the scalar and vector linear solvability of networks and representations of matroids and discrete polymatroids follow as special cases. An algorithm is provided to construct networks which admit FNC solution over $\mathbb{F}_q,$ from discrete polymatroids representable over $\mathbb{F}_q.$  Example networks constructed from discrete polymatroids using the algorithm are provided, which do not admit any scalar and vector solution, and for which FNC solutions with the message dimensions being different provide a larger throughput than FNC solutions with the message dimensions being equal.  
\end{abstract}
\section{Introduction and Background}
Network coding is a technique in which intermediate nodes combine packets before forwarding them, instead of simply routing the packets. In \cite{Ah}, Ahlswede et. al. showed that there exists networks which do not admit any routing solution, but admit scalar linear network coding solutions. In \cite{Li}, it was shown that for multicast networks, scalar linear solutions exist for sufficiently large field size. An algebraic framework for finding linear solutions in networks was introduced in \cite{KoMe}.

It was shown in \cite{MeEfHoKa} that there exists networks which do not admit any scalar linear solution over $\mathbb{F}_q,$ but admit vector linear solution over $\mathbb{F}_q.$ In scalar and vector  network coding, it is inherently assumed that the dimensions of the message vectors are the same and it is also the same as the dimensions of the vectors carried in the edges of the network.  
It is possible that a network does not admit any scalar or vector solution, but admits a solution if all the dimensions of the message vectors are not equal to the edge vector dimension. Such network coding solutions, called Fractional Network Coding (FNC) solutions have been considered in \cite{CaDoFrZe,KiMe_FNC,DoFrZe_FNC}. The work in \cite{CaDoFrZe} primarily focusses on fractional routing, which is a special case of FNC. 
In \cite{KiMe_FNC}, algorithms were provided to compute the capacity region for a network, which was defined to be the closure of all rates achievable using FNC. In \cite{DoFrZe_FNC}, achievable rate regions for certain specific networks were found and it was shown that achievable rate regions using linear FNC need not be convex.

In \cite{DoFrZe}, the connection between scalar linear network coding and representable matroids was established. It was shown in \cite{DoFrZe} that if a scalar linear solution over $\mathbb{F}_q$ exists for a network, then the network is matroidal with respect to a matroid representable over $\mathbb{F}_q.$ The converse that a scalar linear solution exists for a network if the network is matroidal with respect to a matroid representable over $\mathbb{F}_q$ was shown in \cite{KiMe}. A procedure to construct networks from matroids was provided in \cite{DoFrZe}, using which it was shown in \cite{DoFrZe_In} that there exists networks which are solvable but are not scalar or vector linearly solvable. The relationship between network coding, index coding and representations of matroids was analyzed in \cite{RoSpGe}. In \cite{PrRa_ISIT}, the notion of matroidal networks introduced in \cite{DoFrZe} was extended to networks with error correction capability and it was shown that a network admits a scalar linear
error correcting network code if and only if it is a matroidal error correcting network associated with a representable matroid. In \cite{PrRa_ISIT,PrRa_ISITA}, networks with error correction capability were constructed from matroids. 

Discrete polymatroids, introduced by  Herzog and Hibi in \cite{HeHi}, are the multi-set analogue of matroids. In our recent work \cite{VvR_ISIT13_arXiv}, the notion of a discrete polymatroidal network was introduced and it was shown that a vector linear solution over $\mathbb{F}_q$ exists for a network if and only if it is discrete polymatroidal with respect to a discrete polymatroid representable over $\mathbb{F}_q.$ In this paper, we provide a more general definition of a discrete polymatroidal network and establish the connection between the representability over $\mathbb{F}_q$ of discrete polymatroids and linear FNC.  
The contributions of this paper are as follows:
\begin{itemize}
\item
The notion of a \textit{$(k_1,k_2,\dotso,k_m;n)$-discrete polymatroidal network} is introduced. For a network in which $m$ message vectors are generated, it is shown that an FNC solution with the $m$ message vectors  dimensions being $k_1,k_2,\dotso k_m$ and the edge vector dimension being $n$ exists if and only if the network is $(k_1,k_2,\dotso,k_m;n)$-discrete polymatroidal with respect to a discrete polymatroid representable over $\mathbb{F}_q.$
\item
The algorithm introduced in \cite{VvR_ISIT13_arXiv} to construct vector linear solvable networks from representable discrete polymatroids, is generalized to obtain networks which admit linear FNC solutions.
\item
 Example networks constructed from discrete polymatroids are provided, which do not admit any scalar and vector solution, and for which FNC solutions with the message dimensions being different provide a larger throughput than FNC solutions for which the message dimensions are the same.   
\end{itemize}
\textbf{\textit{Notations:}}
The set $\lbrace 1,2,\dotso,r \rbrace$ is denoted as  $\lceil r \rfloor.$ $\mathbb{Z}_{\geq 0} $ denotes the set of non-negative integers. For a vector $v$ of length $r$ and $A \subseteq \lceil r \rfloor,$ $v(A)$ is the vector obtained by taking only the components of $v$ indexed by the elements of $A.$ The $r$ length vector whose $i^{\text{th}}$ component is one and all other components are zeros is denoted as $\epsilon_{i,r}.$  For $u,v \in \mathbb{Z}_{\geq 0}^r,$ $u \leq v$ if all the components of $v-u$ are non-negative and,  $u < v$ if $u \leq v$ and $u \neq v.$ For $u,v \in \mathbb{Z}_{\geq 0}^r,$ $u \vee v$ is the vector whose $i^{\text{th}}$ component is the maximum of the $i^{\text{th}}$ components of $u$ and $v.$ A vector $u \in \mathbb{Z}_{\geq 0}^r$ is called an integral sub-vector of $v \in \mathbb{Z}_{\geq 0}^r$ if $u < v.$ For a set $A,$ $\vert A \vert$ denotes its cardinality and for a vector $v \in  \mathbb{Z}_{\geq 0}^r,$ $\vert v \vert$ denotes the sum of its components. 

\section{Preliminaries}
\subsection{Fractional Network Coding: Definitions and Notations}
A communication network consists of a directed, acyclic graph with the set of vertices denoted by $\mathcal{V}$ and the set of edges denoted by $\mathcal{E}.$ For an edge $e$ directed from $x$ to $y,$ $x$ is called the head vertex of $e$ denoted by $head(e)$ and $y$ is called the tail vertex of $e$ denoted by $tail(e).$ The in-degree of an edge $e$ is the in-degree of its head vertex and out-degree of $e$ is the out-degree of its tail vertex. The messages in the network are generated at edges with in-degree zero, which are called the input edges of the network and let $\mathcal{S}  \subset \mathcal{E}$ denote the set of input edges with $\vert \mathcal{S} \vert =m.$ Let $x_i, i \in \lceil m \rfloor,$ denote the row vector of length $k_i$ generated at the $i^{\text{th}}$ input edge of the network. Let $x=[ x_1,x_2,\dotso,x_m ].$ An edge which is not an input edge is referred to as an intermediate edge. All the intermediate edges in the network are assumed to carry a vector of dimension $n$ over $\mathbb{F}_q.$ A vertex $v \in \mathcal{V}$ demands the set of messages generated at the input edges given by $\delta(v) \subseteq \mathcal{S},$ where $\delta$ is called the demand function of the network. $In(v)$ denotes the set of incoming edges of a vertex $v $ ($In(v)$ includes the intermediate edges as well as the input edges which are incoming edges at node $v$) and $Out(v)$ denotes the union of the set of intermediate edges originating from $v$ and $\delta(v).$

A $(k_1,k_2,\dotso,k_m;n)$-FNC solution over $\mathbb{F}_q$ is a collection of functions $\lbrace \psi_e : \mathbb{F}_q^{\sum_{i=1}^{m}k_i} \rightarrow \mathbb{F}_q^{k_i} , e \in \mathcal{S}\rbrace \cup \lbrace \psi_e : \mathbb{F}_q^{\sum_{i=1}^{m}k_i} \rightarrow \mathbb{F}_q^{n} , e \in \mathcal{E}\setminus\mathcal{S}\rbrace,$ where the function $\psi_e$ is called the global encoding function associated with the edge $e.$  The global encoding functions satisfy the following conditions:
\begin{description}
\item [(N1):]
$\psi_i(x)=[x_i], \forall i \in \mathcal{S},$
\item [(N2):]
For every $v \in \mathcal{V},$ for all $j \in \delta(v),$ there exists a function $\chi_{v,j} : \mathbb{F}_{q}^{n\vert In(v)\vert} \rightarrow \mathbb{F}_q^{k_j}$ called the decoding function for message $j$ at node $v$ which satisfies $\chi_{v,j}(\psi_{i_1}(x),\psi_{i_2}(x), \dotso, \psi_{i_t(x)})=x_{j},$  where $In(v)=\{i_1,i_2, \dotso i_t \} .$
\item [(N3):]
For all \mbox{$i \in \mathcal{E} \setminus \mathcal{S},$} there exists \mbox{$\phi_i: {\mathbb{F}_q}^{n \vert In(head(i))\vert} \rightarrow \mathbb{F}_q^n$} such that \mbox{$\psi_i(x)=\phi_i(\psi_{i_1}(x),\psi_{i_2}(x), \dotso, \psi_{i_r}(x)),$}  where $In(head(i))=\{i_1,i_2, \dotso i_r \}.$ The function $\phi_i$ is called the local encoding function associated with edge $i.$ 
\end{description}

Note that the dimension of the $i^{th}$ message vector $k_i$ need not necessarily be lesser than the edge vector dimension $n.$ For example, as shown in Fig. \ref{frac_routing}, for the network considered, a (2,2;1)-FNC solution exits which is in fact a fractional routing solution.

\begin{figure}[t]
\centering
\includegraphics[totalheight=1.75 in,width=1in]{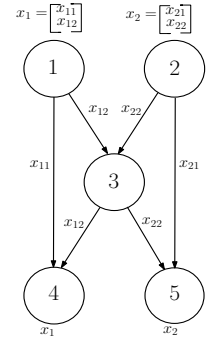}
\caption{A network which admits a $(2,2;1)$-FNC solution}
\label{frac_routing}
\end{figure}   

An FNC solution with \mbox{$k_1=k_2=\dotso=k_m=n=1$} reduces to a scalar solution and an FNC solution for which \mbox{$k_1=k_2=\dotso=k_m=n=k$} reduces to a vector solution of dimension $k.$ 
A solution for which all the local encoding functions and hence the global encoding functions are linear is said to be a linear solution. For a linear $(k_1,k_2,\dotso,k_m;n)$-FNC solution, the global encoding function $\psi_i, i \in \mathcal{E},$ is of the form $\psi_i(x)= x M_i,$ where $M_i$ is an $\sum_{i=1}^{m}k_i \times n$ matrix over $\mathbb{F}_q$ called the global encoding matrix associated with edge $i.$

If a network admits a $(k_1,k_2,\dotso,k_m;n)$-FNC solution, then $(k_1/n,k_2/n,\dotso,k_m/n)$ is said to be an achievable rate vector and the scalar $\frac{1}{m}\sum_{i=1}^m \frac{k_i}{n}$ is said to be an achievable average rate \cite{DoFrZe_FNC}. The closure of the set of all achievable rate vectors is said to be the achievable rate region of the network and the supremum of all achievable average rates is said to be the average coding capacity of the network \cite{DoFrZe_FNC}. 

In this paper, for the sample networks considered, we evaluate the average throughput advantage provided by the FNC solutions which allow different message vector dimensions over the FNC solutions in which the  message vector dimensions are assumed to be the same. Towards this, we call a $(k,k,\dotso,k;n)$ FNC solution to be a symmetric FNC solution and we say that the scalar $k/n$ is a symmetric achievable rate. We define the supremum of all symmetric achievable rates to be the symmetric coding capacity of the network.   
\subsection{Discrete Polymatroids and Matroids}
In this subsection, basic definitions related to discrete polymatroids, matroids and their representability are presented. For a comprehensive treatment of discrete polymatroids and matroids, interested readers are referred to \cite{We,Ox,HeHi,Vl}. For examples illustrating the connection between the vector linear solvability of networks and representations of discrete polymatroids and matroids, see \cite{VvR_ISIT13_arXiv}.
\subsubsection{Discrete Polymatroids}
\begin{definition}[\cite{HeHi}]
Let $\mathbb{D}$ be a non-empty finite set of vectors in $\mathbb{Z}_{\geq 0}^r,$
which contains with each $u \in \mathbb{D}$ all its integral sub-vectors. The set $\mathbb{D}$ is called
a discrete polymatroid on the ground set $\lceil r \rfloor$ if for all $u, v \in \mathbb{D}$ with $\vert u\vert < \vert v\vert,$
there is a vector $w \in  \mathbb{D}$ such that $u < w  \leq  u \vee v.$
\end{definition}

The function $\rho^{\mathbb{D}}: 2^{\lceil r \rfloor} \rightarrow \mathbb{Z}_{\geq 0}$ called the rank function of $\mathbb{D}$ is defined as $\rho^{\mathbb{D}}(A)=\max \{ \vert u(A) \vert , u \in \mathbb{D}\},$ where $\phi \neq A  \subseteq \lceil r \rfloor$ and $\rho^{\mathbb{D}}(\phi)=0.$ In terms of the rank function $\rho^{\mathbb{D}},$ the discrete polymatroid can be written as  $\mathbb{D}=\lbrace x \in \mathbb{Z}_{\geq 0}^r: \vert x(A) \vert \leq \rho^{\mathbb{D}}(A), \forall A \subseteq \lceil r \rfloor \rbrace.$  For simplicity, in the rest of the paper, the rank function of $\mathbb{D}$ is denoted as $\rho.$

From Proposition 4 in \cite{FaFaPa}, it follows that the a function $\rho: 2^{\lceil r \rfloor} \rightarrow \mathbb{Z}_{\geq 0}$ is the rank function of a discrete polymatroid if and only if it satisfies the conditions,
\begin{description}
\item [(D1)]
If $A \subseteq B \subseteq \lceil r \rfloor,$ then $\rho(A)\leq \rho(B).$
\item [(D2)]
 $\forall A,B \subseteq \lceil r \rfloor,$  $\rho(A \cup B) + \rho(A\cap B)\leq\rho(A)+\rho(B).$
\item [(D3)]
$\rho(\phi)=0.$
\end{description}
 A vector $u \in \mathbb{D}$ is a basis vector of $\mathbb{D},$ if $u <v$ for no $v \neq u \in \mathbb{D}.$ The set of basis vectors of $\mathbb{D}$ is denoted as $\mathcal{B}(\mathbb{D}).$ For all $u \in \mathcal{B}(\mathbb{D}),$ $\vert u \vert$ is equal \cite{Vl}, which is called the rank of $\mathbb{D},$ denoted by $rank(\mathbb{D}).$  
%

Let $E$  be a vector space over $\mathbb{F}_q$ and $V_1,V_2,\dotso, V_r$ be finite dimensional vector subspaces of $E.$ Let the mapping $\rho: 2^{\lceil r \rfloor} \rightarrow \mathbb{Z}_{\geq 0}$ be defined as $\rho(X)=dim(\sum_{i \in X} V_i), X \subseteq \lceil r \rfloor.$ It can be verified that $\rho$ satisfies (D1)--(D3) and is the rank function of a discrete polymatroid, denoted by $\mathbb{D}(V_1,V_2,\dotso,V_r).$ 
Note that $\rho$ remains the same even if we replace the vector space $E$ by the sum of the vector subspaces $V_1,V_2,\dotso,V_r.$ In the rest of the paper, the vector subspace $E$ is taken to be the sum of the  vector subspaces $V_1,V_2,\dotso,V_r$ considered. 
The vector subspaces $V_1,V_2,\dotso,V_r$ can be described by a matrix $A=[A_1 \; A_2 \; \dotso A_r ],$ where $A_i, i \in \lceil r \rfloor,$ is a matrix whose columns span $V_i.$
\begin{definition}[\cite{FaFaPa}]
A discrete polymatroid $\mathbb{D}$ is said to be representable over $\mathbb{F}_q$ if there exists vector subspaces $V_1,V_2,\dotso,V_r$ of a vector space $E$ over $\mathbb{F}_q$ such that $dim(\sum_{i \in X} V_i)=\rho(X), \forall X \subseteq \lceil r \rfloor.$ The set of vector subspaces $V_i,i\in\lceil r \rfloor,$ is said to form a representation of $\mathbb{D}.$
\end{definition}



Examples of representable discrete polymatroids are provided in the following two examples.
\begin{example}
Let $A= \underbrace{\hspace{-.3 cm}\left[\begin{matrix}\hspace{.2 cm} 1 \\ \hspace{.2 cm}0 \\\hspace{.2 cm}0 \end{matrix}\right.}_{A_1} \; \underbrace{\begin{matrix} 0\\ 1\\0 \end{matrix}}_{A_2} \; \underbrace{\begin{matrix} 0\\ 0\\1 \end{matrix}}_{A_3}  \; \underbrace{\left.\begin{matrix} 1&0\\ 0&1\\0&1 \end{matrix}\right]}_{A_4}$ be a matrix over $\mathbb{F}_q.$ Let $V_i$ denote the column span of $A_i,$ \mbox{$i \in \lceil 4 \rfloor.$} The rank function $\rho$ of the discrete polymatroid $\mathbb{D}(V_1,V_2,V_3,V_4)$ is as follows: \mbox{$\rho(X)=1, \text{\;if\;} X \in \left\{\{1\},\{2\},\{3\}\right\};$} \mbox{$\rho(X)=2, \text{\;if\;} X \in \left\{\{1,2\},\{1,3\},\{1,4\},\{2,3\},\{4\}\right\}$} and \mbox{$\rho(X)=3 \text{\; otherwise}.$} The set of basis vectors for this discrete polymatroid is given by, {\small$\left\{(0,0,1,2),(0,1,0,2),(0,1,1,1),(1,0,1,1),(1,1,0,1),(1,1,1,0)\right\}.$}
\end{example} 

\begin{example}
Let $A= \underbrace{\hspace{-.0 cm}\left[\begin{matrix} 1 & 0\\ 0 & 1\\0&0\\0&0 \end{matrix}\right.}_{A_1} \; \underbrace{\begin{matrix} 0\\ 0\\1\\0 \end{matrix}}_{A_2} \; \underbrace{\begin{matrix} 0\\ 0\\0\\1 \end{matrix}}_{A_3}  \; \underbrace{\left.\begin{matrix} 1&1\\ 1&0\\1&1\\1&0 \end{matrix}\right.}_{A_4}\; \underbrace{\left.\begin{matrix} 0&0\\ 0&1\\0&1\\1&0 \end{matrix}\right]}_{A_5}$ be a matrix over $\mathbb{F}_q.$ Let $V_i$ denote the column span of $A_i,$ \mbox{$i \in \lceil 5 \rfloor.$} Then the rank function $\rho$ of the discrete polymatroid $\mathbb{D}(V_1,V_2,V_3,V_4,V_5)$ is as follows: \mbox{$\rho(X)=1, \text{\;if\;} X \in \left\{\{2\},\{3\}\right\};$} \mbox{$\rho(X)=2, \text{\;if\;} X \in \left\{\{1\},\{4\},\{5\},\{2,3\},\{3,5\}\right\};$} \mbox{$\rho(X)=3, \text{\;if\;} X \in \left\{\{1,2\},\{1,3\},\{2,4\},\{2,5\},\{3,4\},\{2,3,5\}\right\}$} and  \mbox{$\rho(X)=4, \text{\;otherwise.}$}
\end{example}
\subsubsection{Matroids}
\begin{definition}[\cite{We}]
A matroid is a pair $(\lceil r \rfloor, \mathcal{I}),$ where $\mathcal{I}$ is a collection of  subsets of $\lceil r \rfloor$ satisfying the following three axioms:
\begin{itemize}
\item
$\phi \in \mathcal{I}.$
\item
If $X \in \mathcal{I}$ and $Y \subseteq X,$ then $Y \in \mathcal{I}.$
\item
If $U,V$ are members of $\mathcal{I}$ with $\vert U \vert =\vert V \vert +1$ there exists $x \in U \setminus V$ such that $V \cup x \in \mathcal{I}.$
\end{itemize} 
\end{definition}

A subset of $\lceil r \rfloor$ not belonging to $\mathcal{I}$ is called a dependent set. A maximal independent set is called a basis set and a minimal dependent set is called a circuit. The rank function of a matroid $\Upsilon: 2^{\lceil r \rfloor} \rightarrow \mathbb{Z}_{\geq 0}$ is defined by
$\Upsilon(A)=\max\{\vert X \vert: X \subseteq A, X \in \mathcal{I} \},$ where $A \subseteq \lceil r \rfloor.$ The rank of the matroid $\mathbb{M},$ denoted by $rank(\mathbb{M})$ is equal to $\Upsilon(\lceil r \rfloor).$

 A function $\Upsilon: 2^{\lceil r \rfloor} \rightarrow \mathbb{Z}_{\geq 0}$ is the rank function of a matroid if and only if it satisfies the conditions (D1)--(D3) and the additional condition that $\Upsilon(X)\leq \vert X \vert, \forall X \subseteq \lceil r \rfloor$ (follows from Theorem 3 in Chapter 1.2 in \cite{We}). Since the rank function of $\mathbb{M}$ satisfies (D1)--(D3), it is also the rank function of a discrete polymatroid denoted as $\mathbb{D}(\mathbb{M}).$ In terms of the set of independent vectors $\mathcal{I}$ of $\mathbb{M},$  the discrete polymatroid $\mathbb{D}(\mathbb{M})$ can be written as $\mathbb{D}(\mathbb{M})=\{\sum_{i \in I} \epsilon_{i,r} :I \in \mathcal{I}\}.$ 

A matroid $\mathbb{M}$ is said to be representable over $\mathbb{F}_q$ if there exists one-dimensional vector subspaces $V_1,V_2, \dotso V_r$ of a vector space $E$ such that $\dim(\sum_{i \in X} V_i)=\Upsilon(X), \forall X \subseteq \lceil r \rfloor$ and the set of vector subspaces $V_i, i \in \lceil r  \rfloor,$ is said to form a representation of $\mathbb{M}.$ The one-dimensional vector subspaces $V_i, i \in \lceil r \rfloor,$ can be described by a matrix $A$ over $\mathbb{F}_q$ with $n$ columns whose $i^{\text{th}}$ column spans $V_i.$ 
It is clear that the set of vector subspaces $V_i, i \in \lceil r  \rfloor,$ forms a representation of $\mathbb{M}$ if and only if it forms a representation of $\mathbb{D}(\mathbb{M}).$ 
\section{Linear FNC and  Discrete Polymatroid Representation}
We define a $(k_1,k_2,\dotso,k_m;n)$-discrete polymatroidal network as follows:
\begin{definition}
\label{defn_DPMN}
A network is said to be $(k_1,k_2,\dotso,k_m;n)$-discrete polymatroidal with respect to a discrete polymatroid $\mathbb{D},$ if there exists a map $f: \mathcal{E} \rightarrow \lceil r \rfloor$ which satisfies the following conditions:
\begin{description}
\item [(DN1):]
$f$ is one-to-one on the elements of $\mathcal{S}.$
\item [(DN2):]
\mbox{\small$\sum_{i \in f(\mathcal{S})} k_i \epsilon_{i,r}\in \mathbb{D}.$}
\item [(DN3):]
\mbox{\small$\forall i \in f(\mathcal{S}),$} \mbox{\small$\rho(\{i\})=k_i$} and \mbox{\small$\displaystyle{\max_{i\in f(\mathcal{E}), i \notin f(\mathcal{S})} \rho(\{i\})=n}.$}
\item [(DN4):]
\mbox{\small$\rho(f(In(x)))=\rho(f(In(x) \cup Out(x))), \forall x \in \mathcal{V}.$}
\end{description}
\end{definition}
For a discrete polymatroid $\mathbb{D},$ let $\rho_{max}(\mathbb{D})=\max_{i \in \lceil r \rfloor} \rho(\lbrace i \rbrace).$
\begin{definition}
\label{defn_DPMN1}
A $(k,k\dotso,k;k)$-discrete polymatroidal network with respect to a discrete polymatroid $\mathbb{D},$ with $k=\rho_{max}(\mathbb{D}),$ is said to be discrete polymatroidal with respect to $\mathbb{D}.$
\end{definition}

The notion of a matroidal network was introduced in \cite{DoFrZe}. It can be verified that a network is matroidal with respect to a matroid $\mathbb{M}$ if and only if it is $(1,1,\dotso,1;1)$-discrete polymatroidal with respect to $\mathbb{D}(\mathbb{M}).$ Note that the definition of a discrete polymatroidal network provided in Definition \ref{defn_DPMN1} is equivalent to the definition of a discrete polymatroidal network provided in \cite{VvR_ISIT13_arXiv}.

From the results in \cite{DoFrZe} and \cite{KiMe}, it follows that a network has scalar linear solution over $\mathbb{F}_q$ if and only if the network is matroidal with respect to a matroid representable over $\mathbb{F}_q.$ In \cite{VvR_ISIT13_arXiv}, it was shown that a network has a $k$-dimensional vector linear solution over $\mathbb{F}_q$ if and only if it is discrete polymatroidal with respect to a representable discrete polymatroid with $\rho_{max}(\mathbb{D})=k.$

In the following theorem, we provide a generalization of these results for FNC. 
\begin{theorem}
\label{thm1}
A network has a $(k_1,k_2,\dotso,k_m;n)$-FNC solution over $\mathbb{F}_q,$ if and only if it is $(k_1,k_2,\dotso,k_m;n)$-discrete polymatroidal with respect to a discrete polymatroid  $\mathbb{D}$ representable over $\mathbb{F}_q.$
\begin{proof}
Let the edge set $\mathcal{E}$ of the network be $\lceil l \rfloor$ and let the message set $\mathcal{S}$ be $\lceil m \rfloor.$ The edges are assumed to be arranged in the ancestral ordering which exists since the networks considered are acyclic and the set of intermediate edges in the network is assumed to be $\{m+1,m+2,\dotso l\}.$ We first prove the `if' part of the theorem. Assume that the network considered is $(k_1,k_2,\dotso,k_m;n)$-discrete polymatroidal with respect to a representable discrete polymatroid $\mathbb{D}(V_1,V_2,\dotso,V_r)$ on the ground set $\lceil r \rfloor.$ For brevity, the discrete polymatroid $\mathbb{D}(V_1,V_2,\dotso, V_r)$ is denoted as $\mathbb{D}.$ Let $f$ denote the network-discrete polymatroid mapping. Since, $f$ is one-to-one on the elements of $\mathcal{S},$ assume $f(\mathcal{S})$ to be $\lceil m \rfloor.$

Without loss of generality, the set $\lceil r \rfloor$ can be taken to be the image of the map $f.$ Otherwise, if the image of the map $f$ is $\{i_1,i_2,\dotso i_t\},$ then the network is $(k_1,k_2,\dotso,k_m;n)$-discrete polymatroidal with respect to the discrete polymatroid $\mathbb{D}(V_{i_1},V_{i_2},\dotso,V_{i_t}),$ with the same network-discrete polymatroid mapping $f.$ (DN1), (DN3) and (DN4) follow from the fact that the network is discrete polymatroidal with respect to $\mathbb{D}(V_1,V_2,\dotso,V_r).$ To show that (DN2) is satisfied, it needs to be shown that the vector $u = \sum_{i \in \lceil m \rfloor} k_i \epsilon_{i,t} \in \mathbb{D}(V_{i_1},V_{i_2},\dotso,V_{i_t}).$ Let $v$ denote the vector defined as $\sum_{i \in \lceil m \rfloor} k_i \epsilon_{i,r}.$ Since, the network is discrete polymatroidal with respect to $\mathbb{D}(V_1,V_2,\dotso,V_r),$ from (DN2), we have,
\begin{equation}
\label{eqn_thm1_proof}
\vert v(A) \vert \leq dim\left(\sum_{j \in A} V_j\right), \forall A \subseteq \lceil r \rfloor.
\end{equation}
To show that $u \in \mathbb{D}(V_{i_1},V_{i_2},\dotso,V_{i_t}),$ it needs to be shown that $\vert u (A)\vert \leq dim(\sum_{j \in A} V_{i_j}), \forall A \subseteq \lceil t \rfloor$ which follows from \eqref{eqn_thm1_proof} and from the fact that any subset of $\{i_1,i_2,\dotso,i_t\}$ is also a subset of $\lceil r \rfloor.$ 

Next it will be shown that $dim(\sum_{i=1}^r V_i)=\sum_{i=1}^m k_i.$ Define $s_0=\lceil m \rfloor$ and $s_1=s_0 \cup \{f(m+1)\}.$ Since the edges are arranged in ancestral ordering, we have $In(head(m+1))\subseteq s_0.$ Hence, from (DN4) we have, $\rho(s_1)=dim(\sum_{i \in s_0}V_i+V_{f(m+1)})=dim(\sum_{i \in s_0}V_i)=\rho(s_0).$ Recursively defining $s_i=s_{i-1} \cup f(m+i),$ it can be shown similarly that $\rho(s_i)=\rho(s_0)=\rho(\lceil m \rfloor).$ For $i=l-m,$ we have $s_{l-m}=\lceil r \rfloor$ and $\rho(s_{l-m})=\rho(\lceil r \rfloor)=\rho(\lceil m \rfloor).$ From (DN2), we have $\sum_{i \in \lceil m\rfloor}k_i \epsilon_{i,r} \in \mathbb{D}.$ Hence from the definition of a discrete polymatroid, we have $\sum_{i=1}^m k_i \leq \rho(\lceil m \rfloor).$ From (D2), we have $\rho(\lceil m \rfloor)\leq \rho(\{1\})+\rho(\{2,3,\dotso,m\})\dotso \leq \sum_{i=1}^m \rho(\{i\}).$ Hence, we have $\rho(\lceil m \rfloor)\leq \sum_{i=1}^m k_i,$ since from (DN3) $\rho(\{i\})=k_i,$ for $i \in f(\mathcal{S}).$ As a result $dim(\sum_{i=1}^r V_i)=\rho(\lceil r \rfloor)=\rho(\lceil m\rfloor)=\sum_{i=1}^m k_i.$ The vector subspace $V_i, i \in \lceil r \rfloor, i \notin \lceil m \rfloor$ can be described by a matrix $A_i$ of size $\sum_{i=1}^m k_i \times n$ whose columns span $A_i.$ For $i \in \lceil m \rfloor,$ the vector subspace $V_i$ can be written as the column span of a matrix $A_i$ of size $\sum_{i=1}^m k_i \times k_i.$ Let $B=[A_1 A_2 \dotso A_m].$ Since $dim({\sum_{i=1}^m}V_i)=\sum_{i=1}^m k_i,$ $B$ is invertible and can be taken to be the $\sum_{i=1}^m k_i \times \sum_{i=1}^m k_i$ identity matrix (Otherwise, it is possible to define $A'_i=B^{-1} A_i$ and $V'_i$ to be the column span of $A'_i$ so that $D(V'_1,V'_2,\dotso,V'_r)=D(V_1,V_2,\dotso,V_r)$).

The claim is that taking the global encoding matrix of edge $i$ to be $A_{f(i)}$ forms a $(k_1,k_2,\dotso,k_m;n)$-FNC solution for the network. The proof of the claim is as follows: Since $B$ is an identity matrix, $A_i x=x_i$ for $i \in \lceil m \rfloor$ and hence (N1) is satisfied. For any node $v$ in the network, from (DN4) it follows that $dim(\sum_{i \in In(v) \cup Out(v)}V_{f(i)})= dim(\sum_{i \in In(v)}V_{f(i)}).$ Hence, $\forall j \in Out(v),$ $A_{f(j)}$ can be written as $\sum_{i \in In(V)} W_i A_{f(i)}.$ Hence, (N2) and (N3) are satisfied. This completes the `if' part of the proof.

For the `only if' part of the proof, assume that the network considered admits a $(k_1.k_2,\dotso,k_m;n)$-FNC solution, with $A_i, i \in \lceil l \rfloor,$ being the global encoding matrix associated with edge $i.$ Consider the discrete polymatroid $D(V_1,V_2,\dotso,V_l),$ where $V_i$ denotes the column span of $A_i.$ Let $f(i)=i, i\in \lceil l \rfloor$ be the mapping from the edge set of the network to the ground set of the discrete polymatroid. It can be verified that the network is $(k_1,k_2,\dotso;n)$-discrete polymatroidal with respect to $\mathbb{D}(V_1,V_2,\dotso,V_l).$ 
\end{proof}
\end{theorem}

The results in \cite{DoFrZe,KiMe} and \cite{VvR_ISIT13_arXiv} on the scalar and vector linear solvability of networks can be obtained as corollaries of Theorem \ref{thm1}, as stated in Corollary 1 and Corollary 2 below. 

\begin{corollary}
A network has a scalar linear solution over $\mathbb{F}_q$ if and only if it is matroidal with respect to a matroid representable over $\mathbb{F}_q.$ 
\end{corollary}

\begin{corollary}
A network has a vector linear solution of dimension $k$ over $\mathbb{F}_q$ if and only if it is discrete polymatroidal with respect to a discrete polymatroid $\mathbb{D}$ representable over $\mathbb{F}_q,$ with $\rho_{max}(\mathbb{D})=k.$  
\end{corollary}

The result in Theorem 1 is illustrated in the following two examples.
\begin{example}
\begin{figure}[t]
\centering
\includegraphics[totalheight=3 in,width=2in]{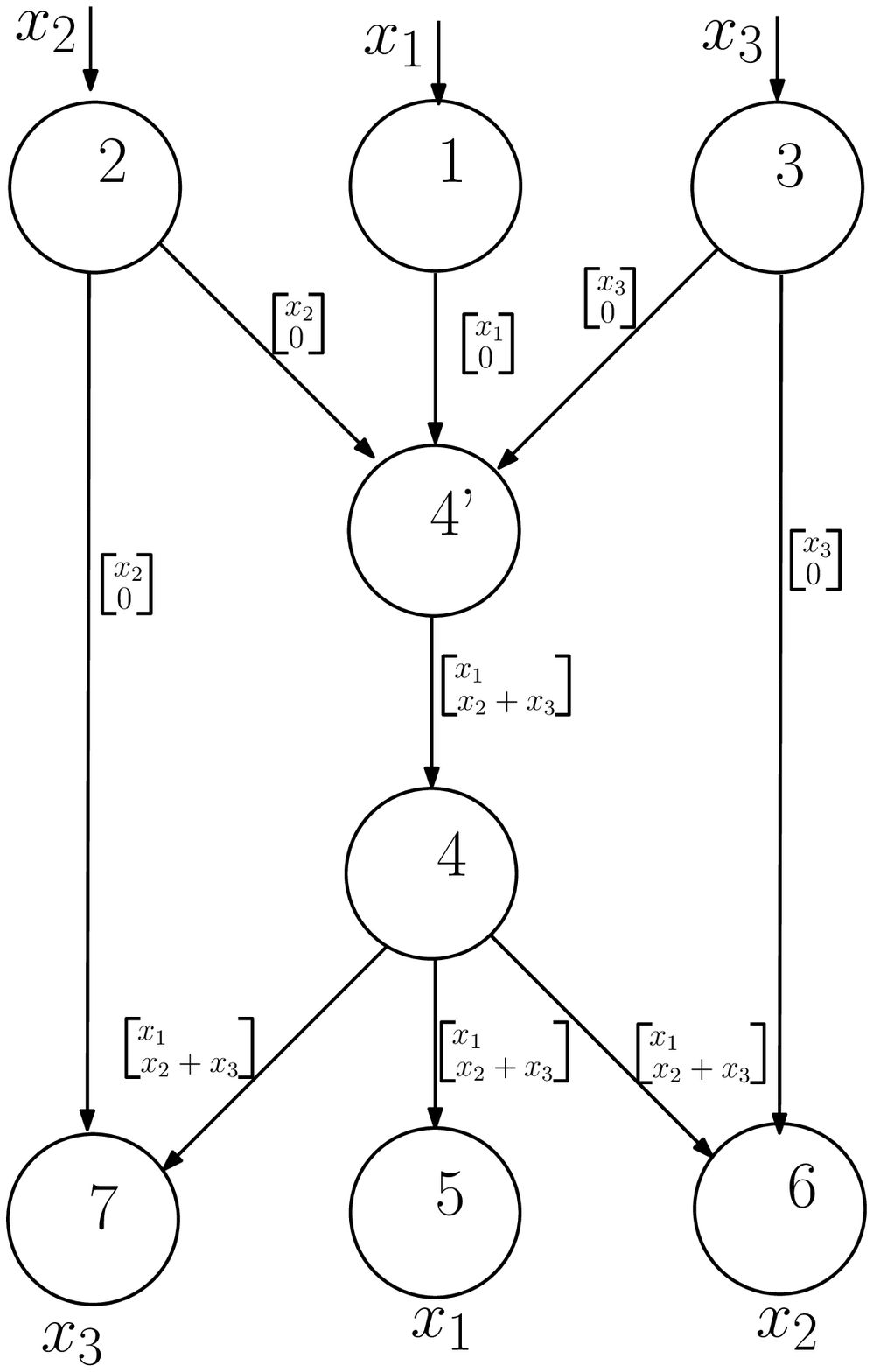}
\caption{A network for which scalar and vector solutions do not exist but an FNC solution exists}
\label{fnc_example}
\end{figure}
Consider the network given in Fig. \ref{fnc_example}. This network admits a linear $(1,1,1;2)$-FNC solution shown in Fig. \ref{fnc_example}. Consider the representable discrete polymatroid $\mathbb{D}(V_1,V_2,V_3,V_4)$ defined in Example 1.  As explained below, the network shown in Fig. \ref{fnc_example} is $(1,1,1;2)$-discrete polymatroidal with respect to the discrete polymatroid $\mathbb{D}(V_1,V_2,V_3,V_4),$ with the network-discrete polymatroid mapping $f$ defined as follows: all the incoming and outgoing edges of node $i, i \in \{1,2,3,4\}$ are mapped on to the ground set element $i$ of the discrete polymatroid $\mathbb{D}(V_1,V_2,V_3,V_4).$ 
\begin{itemize}
\item
Clearly, $f$ is one-to-one on the elements of $\mathcal{S}$ and hence (DN1) is satisfied.  
\item
From Example 1, it follows that the vector \mbox{$\displaystyle{\sum_{i \in \{1,2,3\}} k_i \epsilon_{i,4}=(1,1,1,0)}$} is a basis vector for $\mathbb{D}(V_1,V_2,V_3,V_4).$ Hence, $(1,1,1,0) \in \mathbb{D}(V_1,V_2,V_3,V_4)$ and (DN2) is satisfied. 
\item
From Example 1, it can be seen that $\rho(\{1\})=\rho(\{2\})=\rho(\{3\})=1$ and \mbox{$\displaystyle{\max_{i\in \lceil 4 \rfloor, i \notin \{1,2,3\}} \rho(\{i\})}=\rho(\{4\})=2.$} Hence, (DN3) is satisfied. 
\item
\begin{itemize}
\item
For \mbox{$x \in \{1,2,3,4\},$} since $f(In(x))=f(In(x) \cup Out(x)),$ $\rho(f(In(x)))=\rho(f(In(x)\cup Out(x))).$
\item
For node 4', we have $f(In(4'))=\{1,2,3\}$ and $f(In(4') \cup Out(4'))=\{1,2,3,4\}.$ 
From Example 1, it follows that $\rho(\{1,2,3\})=\rho(\{1,2,3,4\})=3.$
\item
For node 5, we have,$f(In(5))=\{4\}$ and $f(In(5)\cup Out(5))=\{1,4\}.$ From Example 1, it follows that $\rho(\{4\})=\rho(\{1,4\})=2.$ 
\item
For node 6, we have $f(In(6))=\{3,4\}$ and $f(In(6)\cup Out(6))=\{2,3,4\}.$ From Example 1, it can be seen that $\rho(\{3,4\})=\rho(\{2,3,4\})=3.$ 
\item
For node 7, we have $f(In(7))=\{2,4\}$ and $f(In(7)\cup Out(7))=\{2,3,4\}.$ From Example 1, it follows that $\rho(\{2,4\})=\rho(\{2,3,4\})=3.$ 
\end{itemize}
For all the nodes $x$ in the network, $\rho(f(In(x)))=\rho(f(In(x)\cup Out(x))).$ Hence (DN4) is satisfied.
\end{itemize}

The network shown in Fig. \ref{fnc_example} has the properties listed in the following lemma.

\begin{lemma}
The network shown in Fig. \ref{fnc_example} has the following properties:
\begin{enumerate}
\item
The network shown in Fig. \ref{fnc_example} does not admit any scalar or vector solution.
\item
The symmetric coding capacity of the network shown in Fig. \ref{fnc_example} is equal to 1/2. Hence, the (1,1,1;2)-FNC solution provided in Fig. \ref{fnc_example}, which is a symmetric FNC solution, achieves the symmetric coding capacity.
\end{enumerate}
\end{lemma}
\begin{proof}
1) To satisfy the demand of node 7, the edge from 4' to 4 has to carry $x_1,$ which would mean that the demands of the nodes 5 and 6 cannot be met. Hence, the network shown in Fig. \ref{fnc_example} does not admit any scalar and vector solution. 

2) Every $(k,k,k;n)$-FNC solution for this network should satisfy the condition that $\frac{k}{n}\leq \frac{1}{2}.$ The reason for this is as follows: $k$ out of $n$ dimensions of the vector flowing in the edge joining 4' and 4 should carry $x_1$ to satisfy the demand of node 7. The demands of node 5 and node 6 should be met by what is carried in the remaining $n-k$ dimensions. Hence, $n-k$ should be at least $k$ to be able to satisfy the demands of nodes 5 and 6. 
\end{proof} 
\end{example}

In the previous example, a symmetric FNC solution was provided. 
In the next example, we provide a network with a non-symmetric FNC solution and for which the average rate achieved by the FNC solution provided is greater than the symmetric coding capacity.
\begin{example}
\begin{figure}[t]
\centering
\includegraphics[totalheight=2 in,width=3in]{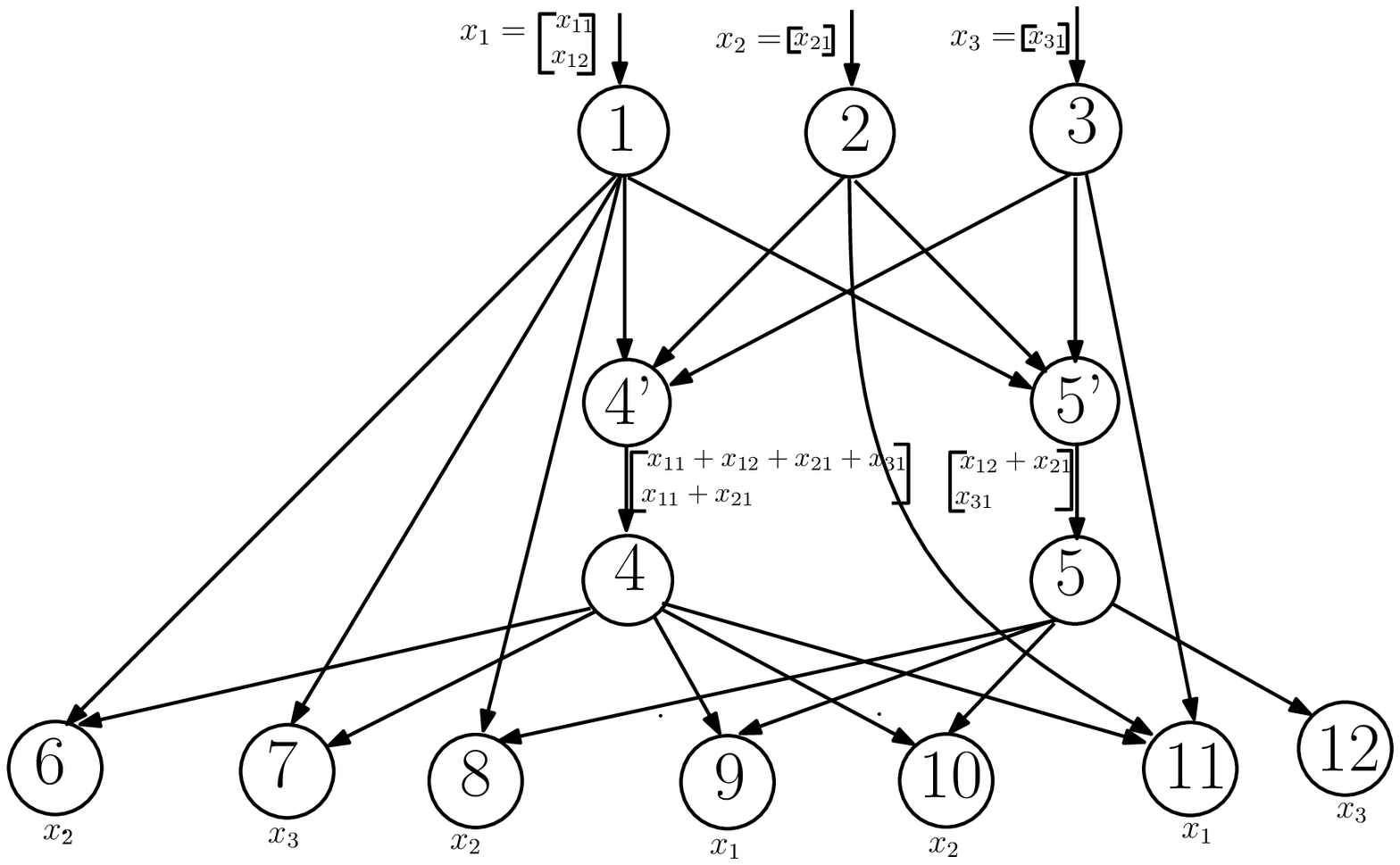}
\caption{A network for which scalar and vector solutions do not exist but an FNC solution exists}
\label{mfnc_example}
\end{figure}
Consider the network given in Fig. \ref{mfnc_example}. A linear (2,1,1;2)-FNC solution for this network is shown in Fig. \ref{mfnc_example}. All the outgoing edges of a node which has only one incoming edge, are assumed to carry the same vector as that of the incoming edge. Consider the discrete polymatroid $D(V_1,V_2,V_3,V_4,V_5)$ defined in Example 2. It can be verified that the network shown in Fig. \ref{mfnc_example} is (2,1,1;2)-discrete polymatroidal with respect to the discrete polymatroid $D(V_1,V_2,V_3,V_4,V_5)$ with the network-discrete polymatroid mapping $f$ defined as follows: all the incoming and outgoing edges of node $i, i \in \{1,2,3,4,5\}$ are mapped on to the ground set element $i$ of the discrete polymatroid $\mathbb{D}(V_1,V_2,V_3,V_4,V_5).$

Lemma 2 below lists some of the properties of the network given in Fig. \ref{mfnc_example}.

\begin{lemma}
The network given in Fig. \ref{mfnc_example} has the following properties:
\begin{enumerate}
\item
The network in Fig. \ref{mfnc_example} does not admit any scalar or vector solution.
\item
The symmetric coding capacity of the network in Fig. \ref{mfnc_example} is 1/2. Hence the (2,1,1;2)-FNC solution provided in Fig. \ref{mfnc_example} achieves an average rate of 2/3 which is strictly greater than the maximum average rate of 1/2 achievable using symmetric FNC. 
\end{enumerate}
\begin{proof}
 1) To deliver message $x_3$ to node 10, the edge connecting nodes 5' and 5 needs to carry $x_3.$ In that case, message $x_2$ cannot be delivered to node 7, since the only path from node 2 which generates $x_2$ to node 7 contains the edge joining 5' and 5. Hence, the network in Fig. \ref{mfnc_example} does not admit any scalar or vector solution.
 
 2) For any $(k,k,k;n)$-FNC solution, $\frac{k}{n}$ cannot exceed $\frac{1}{2}.$ The reason is as follows: $k$ dimensions of the vector transmitted from 5' to 5 should carry $x_3$ and to ensure that node 7 gets $x_2,$ $n-k$ should be at least $k,$ i.e, $\frac{k}{n}\leq \frac{1}{2}.$ The linear $(2,1,1;2)$-FNC solution shown in Fig. \ref{fnc_example} achieves an average rate of 2/3, which is greater than the maximum average rate of 1/2 achievable using a symmetric FNC solution.  
\end{proof}
\end{lemma}
\end{example}

The networks in Example 3 and Example 4 have been constructed from discrete polymatroids using the algorithm provided in the next subsection.
\section{Construction of Networks from Discrete Polymatroids}
In this section, we extend the algorithm provided in \cite{VvR_ISIT13_arXiv} to construct networks from discrete polymatroids. The network constructed admits a linear FNC solution over $\mathbb{F}_q,$ if the discrete polymatroid from which it was constructed is representable over $\mathbb{F}_q.$ Before presenting the algorithm, some useful definitions are provided.
\begin{definition}
For a discrete polymatroid $\mathbb{D},$ a vector $u \in \mathbb{Z}_{\geq 0}^{r}$ is said to be an excluded vector if the $i^{\text{th}}$ component of $u$ is less than or equal to $\rho(\lbrace i \rbrace), \forall i \in \lceil r \rfloor,$ and $u \notin \mathbb{D}.$
\end{definition}

For a discrete polymatroid $\mathbb{D},$ let $\mathcal{D}(\mathbb{D})$ denote the set of excluded vectors. 

For a vector $u \in \mathbb{Z}_{\geq 0}^{r},$ let $(u)_{>0}$ denotes the set of indices corresponding to the non-zero components of $u.$


Let  $\mathcal{D}_i(\mathbb{D}), i \in \lceil r \rfloor$ denote the set of excluded vectors whose $i^{\text{th}}$ component is 1. 

Let $\mathcal{C}_i(\mathbb{D}), i \in \lceil r \rfloor$ denote the set of vectors $u \in \mathcal{D}_i(\mathbb{D})$ which satisfy the following three conditions:
\begin{enumerate}
\item
$u-\epsilon_{i,r} \in \mathbb{D}.$
\item
There does not exist $ v \neq u \in \mathcal{D}_i(\mathbb{D})$ for which  $v < u.$ 
\item
$(v)_{>0} \not \subset (u)_{>0},$ for all $v \neq u \in \mathcal{D}_i(\mathbb{D}).$ 
\end{enumerate} 

\begin{example}
For the discrete polymatroid considered in Example 1, the set of vectors $\mathcal{D}_i(\mathbb{D}),i\in\lceil 4 \rfloor,$ are as given below:

{\vspace{-.5 cm}
\footnotesize
\begin{align*}
&\mathcal{D}_1(\mathbb{D})=\lbrace (1,0,0,2),(1,0,1,2),(1,1,0,2),(1,1,1,1),(1,1,1,2),\rbrace\\
&\mathcal{D}_2(\mathbb{D})=\lbrace (0,1,1,2),(1,1,0,2),(1,1,1,1),(1,1,1,2),\rbrace\\
&\mathcal{D}_3(\mathbb{D})=\lbrace (0,1,1,2),(1,0,1,2),(1,1,1,1),(1,1,1,2),\rbrace\\
&\mathcal{D}_4(\mathbb{D})=\lbrace (1,1,1,1)\rbrace.\\
\end{align*}
\vspace{-.5 cm}}

\noindent The set of vectors \mbox{$\mathcal{C}_i(\mathbb{D}),$} \mbox{$i \in \lceil 4 \rfloor,$} are given by \mbox{$\mathcal{C}_1(\mathbb{D})=\lbrace (1,0,0,2)\rbrace,$}
\mbox{$\mathcal{C}_2(\mathbb{D})=\lbrace (0,1,1,2)\rbrace,$}
\mbox{$\mathcal{C}_3(\mathbb{D})=\lbrace (0,1,1,2)\rbrace$} and
\mbox{$\mathcal{C}_4(\mathbb{D})=\lbrace (1,1,1,1)\rbrace.$}
\end{example}

\begin{example}
For the discrete polymatroid considered in Example 2, it can be verified that the sets $\mathcal{C}_i(\mathbb{D}),i \in \lceil 5 \rfloor$ are given by,
 \mbox{\small$\mathcal{C}_1(\mathbb{D})=\lbrace (1,0,0,2,2),(1,1,1,2,0)\rbrace,$}
\mbox{\small$\mathcal{C}_2(\mathbb{D})=\lbrace (0,1,0,2,2),(2,1,0,0,2),(2,1,0,2,0)\rbrace,$}
\mbox{\small$\mathcal{C}_3(\mathbb{D})=\lbrace (2,0,1,2,0),(0,0,1,0,2)\rbrace,$}
\mbox{\small$\mathcal{C}_4(\mathbb{D})=\lbrace (0,0,1,1,2),(2,1,1,1,0),\rbrace$} and
\mbox{\small$\mathcal{C}_5(\mathbb{D})=\lbrace (0,0,1,2,1),(2,0,0,2,1),(2,1,1,0,1)\rbrace.$}
\end{example}

The algorithm useful towards constructing networks from Discrete polymatroids is as follows:\\
\underline{\textbf{{ALGORITHM 1}}}\\
\underline{\emph{Step 1:}}
 Choose a basis vector $b \in \mathcal{B}(\mathbb{D})$ given by $\sum_{i \in (b)_{>0}} k_i \epsilon_{i,r}$ which satisfies the condition that $\rho(\{i\})=k_i, \forall i \in (b)_{>0}.$ For every $i \in (b)_{>0},$ add a node $i$ to the network with an input edge $e_i$ which generates the message $x_i.$ Let $f(e_i)=i.$ Define $M=T=(b)_{>0}.$\\
\underline{\emph{Step 2:}}
For $i \in \lceil r \rfloor \notin T,$ find a vector $u \in \mathcal{C}_i(\mathbb{D}),$ for which \mbox{$(u-\epsilon_{i,r})_{>0} \subseteq T.$} Add a new node $i'$ to the network with incoming edges from all the nodes which belong to $(u-\epsilon_{i,r})_{>0}.$ Also, add a node $i$ with a single incoming edge from $i',$ denoted as $e_{i',i}.$ Define $f(e)=head(e), \forall e \in In(i)$ and $f(e_{i',i})=i.$ Let $T=T \cup \lbrace i\rbrace.$
Repeat step 2 until it is no longer possible.\\
\underline{\emph{Step 3:}} 
For $i \in M,$ choose a vector $u$ from $\mathcal{C}_i(\mathbb{D})$ for which $(u)_{>0} \subseteq T.$ Add a new node $h$ to the network which demands message $x_i$ and which has connections from the nodes in $(u-\epsilon_{i,r})_{>0}.$ Define $f(e)=head(e), \forall e \in In(h).$  Repeat this step as many number of times as desired. 
%

Theorem 2 below establishes the connection between the network constructed using Algorithm 1 and the discrete polymatroid from which the network was constructed, for a discrete polymatroid representable over $\mathbb{F}_q.$

For a basis vector \mbox{$b \in \mathcal{B}(\mathbb{D}),$} define \mbox{$\displaystyle{\phi(b)=\max_{i \notin (b)_{>0}, i \in \lceil r \rfloor} \rho(\{i\})}.$}
\begin{theorem}
A network constructed using ALGORITHM 1 from a discrete polymatroid $\mathbb{D}$  which is representable over $\mathbb{F}_q,$ with the basis vector $b$ given by $\sum_{i \in (b)_{>0}} k_i \epsilon_{i,r}$ chosen in Step 1 and $\phi(b)=n,$  admits a linear $(k_1,k_2,\dotso,k_m;n)$-FNC solution over $\mathbb{F}_q.$
\begin{proof}
The proof of the theorem is given by showing that the constructed network is $(k_1,k_2,\dotso,k_m;n)$-discrete polymatroidal with respect to the representable discrete polymatroid $D$ from which it is constructed. The satisfaction of (DN1) is ensured by step 1 of the construction procedure. Since the vector $\sum_{i\in \mathcal{S}} k_i \epsilon_{i,r}$ belongs to $\mathcal{B}(\mathbb{D}),$ it belongs to $\mathbb{D}$ as well and hence (DN2) is satisfied. Also, since $\rho(\{i\})=i, \forall i \in (b)_{>0}$ and $\phi(b)=n,$ (DN3) is satisfied.

The nodes in the network constructed using Algorithm 1 are of four kinds (i) node $i,$ $i \in M,$ which are added in step 1 (ii) node $i', i \in \lceil r \rfloor \setminus M,$ (iii) node $i, i \in \lceil r \rfloor \setminus M,$ added in Step 2. (iv) nodes added in Step 3 which demand messages. Following a similar approach as in the proof of Theorem 2 in \cite{VvR_ISIT13_arXiv}, it can be shown that for a node $x$ which belongs to any one of the four kinds, $\rho(f(In(x)))=\rho(f(In(x)\cup Out(x)))$ and hence (DN4) is satisfied.
\end{proof}
\end{theorem}

The following examples illustrate the construction procedure provided in Algorithm 1.
\begin{example}
Continuing with Example 5, the construction procedure for the discrete polymatroid considered in Example 1 is summarized in Table I. The different steps involved in the construction are depicted in Fig. \ref{example1_construction}. Since, in step 1, the basis vector $b=(1,1,1,0)$ is used and $\phi(b)=\rho(\{4\})=2,$ the constructed network admits a linear $(1,1,1;2)$-FNC solution. The linear $(1,1,1;2)$-FNC solution shown in Fig. \ref{example1_construction} is obtained by taking the global encoding matrix of the edge joining 4' and 4 to be the matrix $A_4$ given in Example 1. 
\begin{table}[h]
\centering
\includegraphics[totalheight=2 in,width=3.5 in]{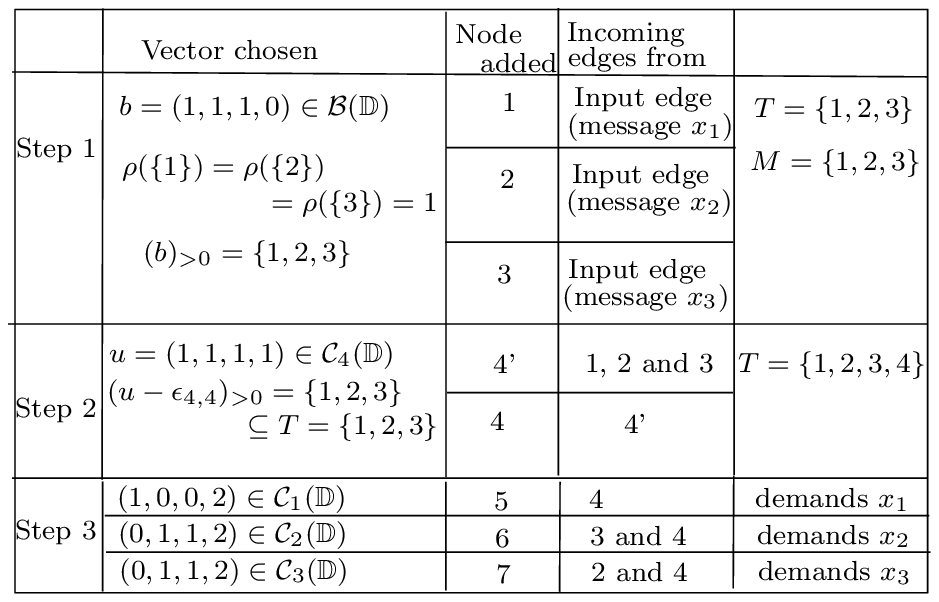}
\caption{Steps involved in the construction of a network from the discrete polymatroid in Example 1}
\label{fig:network1_table}
\end{table} 
\begin{figure}[h]
\centering
\includegraphics[totalheight=2.75 in,width=3.5 in]{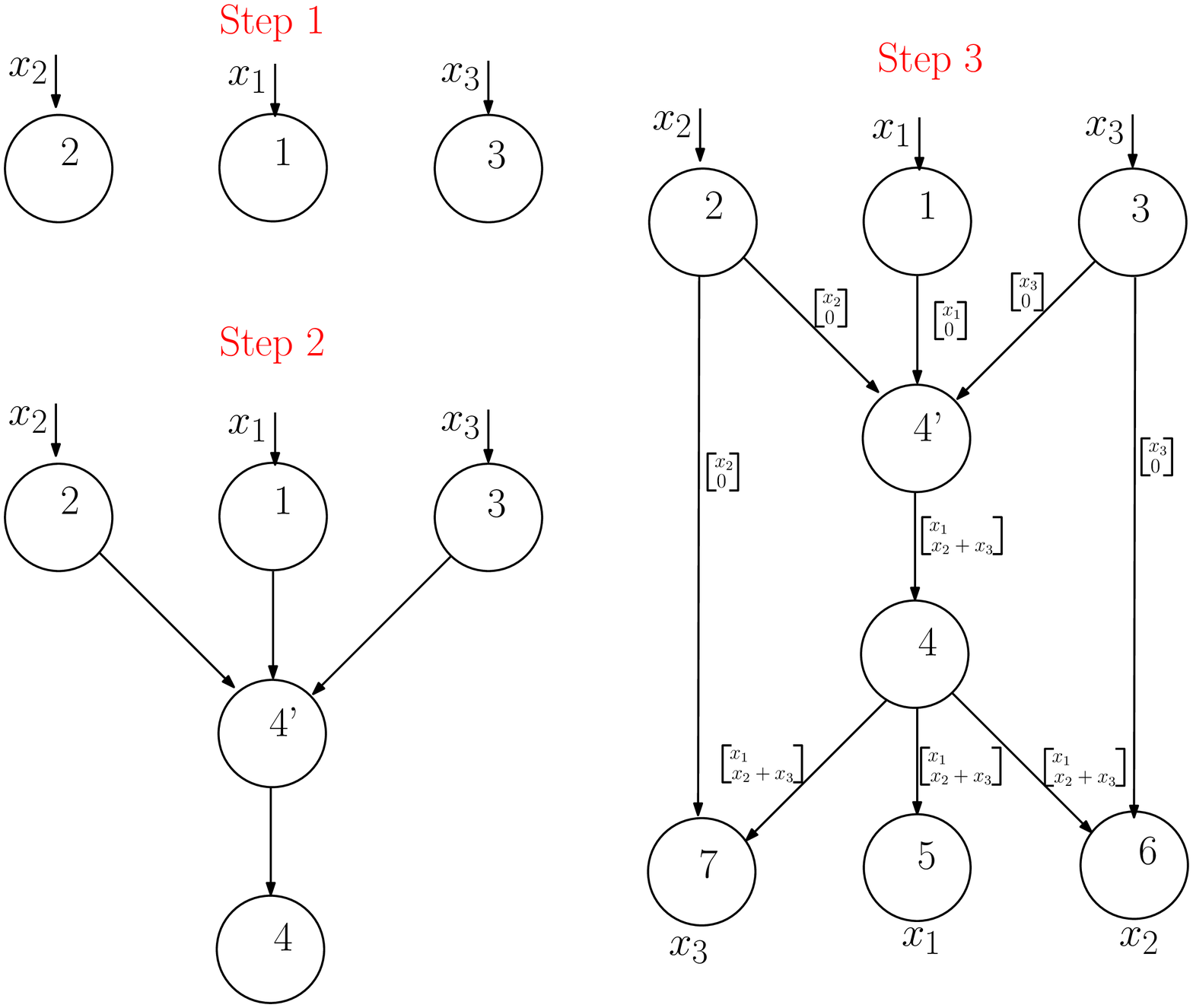}
\caption{Diagram showing the steps involved in the construction of a network from the discrete polymatroid in Example 1}
\label{example1_construction}
\end{figure} 
\end{example}
\begin{example}
Continuing with Example 6, the construction procedure for the discrete polymatroid considered in Example 2 is summarized in Table II. The different steps involved in the construction are depicted in Fig. \ref{example2_construction}. Since, in step 1, the basis vector $b=(2,1,1,0,0)$ is used and $\phi(b)=\rho(\{4\})=\rho(\{5\})=2,$ the constructed network admits a linear $(2,1,1;2)$-FNC solution. The linear $(2,1,1;2)$-FNC solution shown in Fig. \ref{example2_construction} is obtained by taking the global encoding matrix of the edge joining 4' and 4 to be the matrix $A_4$ given in Example 2 and that of the edge joining 5' and 5 to be the matrix $A_5$  given in Example 2.
\begin{table}[h]
\centering
\includegraphics[totalheight=2 in,width=3.5 in]{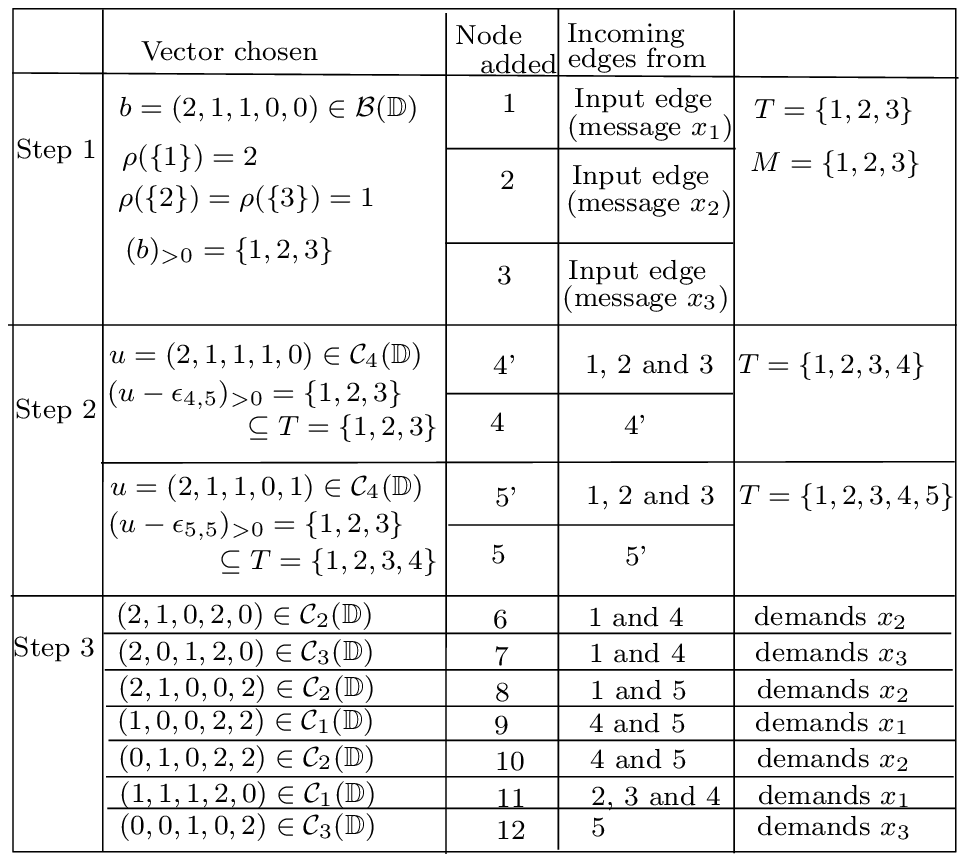}
\caption{Steps involved in the construction of a network from the discrete polymatroid in Example 2}
\label{fig:network2_table}
\end{table} 
\begin{figure}[h]
\centering
\includegraphics[totalheight=3.5 in,width=3.5 in]{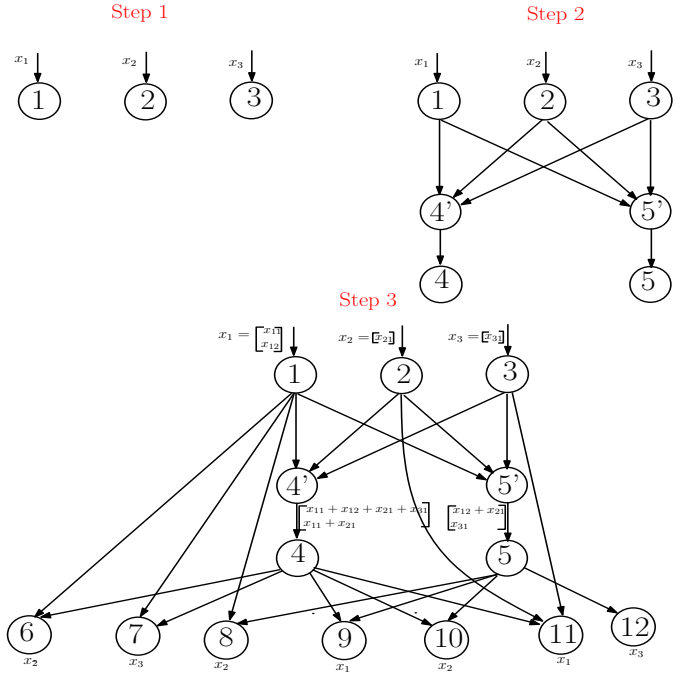}
\caption{Diagram showing the steps involved in the construction of a network from the discrete polymatroid in Example 2}
\label{example2_construction}
\end{figure} 
\end{example}
\section{Conclusion}
The connection between the existence of a linear $(k_1,k_2,\dotso,k_m;n)$-FNC solution for a network over $\mathbb{F}_q$ and the network being $(k_1,k_2,\dotso,k_m;n)$-discrete polymatroidal with respect to a discrete polymatroid representable over $\mathbb{F}_q$ was established. Using the algorithm provided to construct networks from discrete polymatroids, example networks were provided which do not admit any scalar or vector solution, but admit a linear FNC solution.   
   
\end{document}